\documentclass[aps,prb,twocolumn,superscriptaddress,floatfix]{revtex4-2}
\usepackage{amsmath,amssymb,amsthm}
\usepackage{booktabs}
\usepackage{hyperref}

\newtheorem{proposition}{Proposition}
\newtheorem{corollary}{Corollary}

\begin{document}

\title{Belief Propagation Convergence Prediction for Bivariate Bicycle Quantum Error Correction Codes}

\author{Anton Pakhunov}
\affiliation{Independent Researcher}

\begin{abstract}
Decoding Bivariate Bicycle (BB) quantum error correction codes typically requires Belief Propagation (BP) followed by Ordered Statistics Decoding (OSD) post-processing when BP fails to converge. Whether BP will converge on a given syndrome is currently determined only after running BP to completion. We show that convergence can be predicted in advance by a single modulo operation: if the syndrome defect count is divisible by the code's column weight $w$, BP converges with high probability (100\% at $p \leq 0.001$, degrading to 87\% at $p = 0.01$); otherwise, BP fails with probability $\geq 90\%$. The mechanism is structural: each physical data error activates exactly $w$ stabilizers, so a defect count not divisible by $w$ implies the presence of measurement errors outside BP's model space. Validated on five BB codes with column weights $w = 2$, 3, and~4, mod-$w$ achieves AUC = 0.995 as a convergence classifier at $p = 0.001$ under phenomenological noise, dominating all other syndrome features (next best: AUC = 0.52). The false positive rate scales empirically as $O(p^{2.05})$ ($R^2 = 0.98$), confirming the analytical bound from Proposition~2. Among BP failures on mod-$w = 0$ syndromes, 82\% contain weight-2 data error clusters, directly confirming the dominant failure mechanism. We further demonstrate that the prediction is invariant under BP scheduling strategy and decoder variant, including Relay-BP~\cite{muller2025} --- the strongest known BP enhancement for quantum LDPC codes --- and characterize its degradation near the code threshold. These results apply directly to IBM's Gross code $[\![144, 12, 12]\!]$ and Two-Gross code $[\![288, 12, 18]\!]$, targeted for deployment in 2026--2028.
\end{abstract}

\maketitle

\section{Introduction}

IBM's quantum computing roadmap relies on a family of codes known as Bivariate Bicycle (BB) codes~\cite{bravyi2024}. The Gross code $[\![144, 12, 12]\!]$ encodes 12 logical qubits in 144 physical qubits --- an encoding rate twelve times higher than surface codes at comparable distance. The Kookaburra processor targets this code for 2026; Starling targets the Two-Gross code $[\![288, 12, 18]\!]$ for 2028~\cite{ibmroadmap}.

Decoding BB codes is harder than decoding surface codes. Minimum-weight perfect matching (MWPM), the standard decoder for surface codes, does not apply directly to BB codes because their parity check matrices contain \textit{hyperedges}: single error events that trigger more than two stabilizer measurements simultaneously. The prevailing approach is Belief Propagation with Ordered Statistics Decoding (BP+OSD)~\cite{roffe2020}.

BP+OSD suffers from a fundamental inefficiency: whether BP will succeed is unknown until it either converges or exhausts its iteration budget. When BP succeeds, decoding takes approximately 46~$\mu$s. When it fails and OSD is invoked, decoding takes approximately 108~$\mu$s. Every syndrome must pass through BP before this outcome is known.

We show that convergence can be predicted in $O(1)$ time, before BP is invoked.

\section{Background}

\subsection{Bivariate Bicycle Codes}

BB codes are constructed from two polynomials $A$ and $B$ over a cyclic group algebra. The construction is detailed in~\cite{bravyi2024}; the property relevant to this work is the \textit{column weight}~$w$ of the parity check matrix.

A column weight of $w$ means each data qubit participates in exactly $w$ stabilizer measurements. For all BB codes on IBM's roadmap, $w = 3$, arising from 3-term polynomials such as $A = x^3 + y + y^2$.

The consequence is immediate: a single X-type physical error on any qubit triggers exactly 3 Z-stabilizer changes, producing exactly 3 syndrome defects --- without exception.

\begin{table}[t]
\caption{BB codes in IBM's quantum roadmap. All have column weight $w = 3$.}
\begin{tabular}{lcc}
\toprule
Code & Column weight $w$ & IBM target \\
\midrule
$[\![72, 12, 6]\!]$ & 3 & --- \\
$[\![144, 12, 12]\!]$ (Gross) & 3 & Kookaburra 2026 \\
$[\![288, 12, 18]\!]$ (Two-Gross) & 3 & Starling 2028 \\
$[\![360, 12, \leq 24]\!]$ & 3 & --- \\
\bottomrule
\end{tabular}
\end{table}

\subsection{Why BP+OSD Is Slow}

BP performs iterative message-passing on the code's Tanner graph to find a consistent error assignment. For surface codes, BP frequently fails due to the abundance of short cycles that trap messages in oscillatory loops. BB codes have fewer short cycles, so BP performs reasonably well --- but not universally.

When BP fails to converge, OSD takes over. OSD is guaranteed to produce a valid solution but requires Gaussian elimination over the most reliable bits --- an $O(n^3)$ operation. For the Gross code, the resulting latencies are:
\begin{itemize}
\item BP converges: ${\sim}46$~$\mu$s (code-capacity) or ${\sim}100$~$\mu$s (phenomenological)
\item BP fails, OSD invoked: ${\sim}108$~$\mu$s (code-capacity) or ${\sim}300$~$\mu$s (phenomenological)
\item Under phenomenological noise at $p = 0.001$: average ${\sim}109$~$\mu$s (given ${\sim}64\%$ convergence rate)
\end{itemize}

If convergence failure were known in advance, the OSD penalty could be avoided entirely for syndromes where BP will succeed.

\section{The Prediction}

\subsection{The Observation}

The key structural fact is the following:

\textit{Each physical data error activates exactly $w$ stabilizers. Therefore, if the total defect count is not divisible by $w$, the syndrome cannot be produced by data errors alone --- at least one measurement error must be present.}

A measurement error flips exactly one stabilizer outcome without a corresponding data error, contributing exactly 1 defect. BP's Tanner graph models data errors only and has no mechanism to represent measurement errors. When the syndrome requires measurement error contributions for consistency, BP cannot find a satisfying assignment and fails to converge.

This yields the following prediction rule:
\begin{verbatim}
if defect_count % w == 0:
    BP will likely converge
    (100% at p <= 0.001; 87% at p = 0.01)
else:
    BP will fail (probability >= 90%)
\end{verbatim}

For all IBM roadmap BB codes, $w = 3$, so the test reduces to divisibility by~3.

\subsection{Why This Works}

\begin{proposition}[Defect parity]
Under phenomenological noise on a BB code with column weight $w$, the syndrome defect count satisfies
\begin{equation}
\mathrm{defect\_count} \bmod w = |\mathbf{m}| \bmod w
\end{equation}
where $|\mathbf{m}|$ is the number of measurement errors. In particular, the defect count modulo $w$ depends only on measurement errors and is independent of data errors.
\end{proposition}

\begin{proof}
Each data error on qubit $j$ activates exactly $w$ stabilizers (the column weight of $H_z$), contributing $w$ defects. Each measurement error flips exactly one stabilizer outcome, contributing 1 defect. Therefore $\mathrm{defect\_count} = w \cdot |\mathbf{e}| + |\mathbf{m}|$, where $|\mathbf{e}|$ is the number of data errors. Reducing modulo $w$: $\mathrm{defect\_count} \bmod w = |\mathbf{m}| \bmod w$.
\end{proof}

\begin{proposition}[Two failure modes]
Under the conditions of Proposition~1, BP can fail on a $\mathrm{mod}\text{-}w = 0$ syndrome in exactly two ways:

(a) \textit{Measurement-error failures}: $|\mathbf{m}| \geq w$ measurement errors preserve $\mathrm{mod}\text{-}w = 0$ but place the syndrome outside the image of $H_z$. The probability of this event, conditioned on $\mathrm{mod}\text{-}w = 0$, is
\begin{equation}
P(|\mathbf{m}| \geq w \mid |\mathbf{m}| \bmod w = 0) = O(p^w)
\end{equation}
with leading term $\binom{n_\mathrm{meas}}{w} p^w (1-p)^{n_\mathrm{meas} - w}$.

(b) \textit{Data-error failures}: even with $|\mathbf{m}| = 0$, weight-2 data error clusters (two errors on qubits sharing a stabilizer) can create frustrated loops in the Tanner graph that prevent BP convergence. The probability of this event is $O(p^2)$.

The total false positive rate is $O(p^2)$, dominated by data-error failures at low $p$.
\end{proposition}

\textit{Proof of (a).} When $|\mathbf{m}| = 0$, the syndrome lies in the image of $H_z$ and a valid solution exists in BP's model space. Measurement-error failures require $|\mathbf{m}| \geq w$ (the minimum nonzero count preserving $|\mathbf{m}| \bmod w = 0$). Conditioned on $|\mathbf{m}| \bmod w = 0$, this probability has leading term $\binom{n_\mathrm{meas}}{w} p^w / (1-p)^{n_\mathrm{meas}}$, vanishing as $O(p^w)$. For $w = 3$, $n_\mathrm{meas} = 360$: this gives $\approx 0.008$ at $p = 0.001$. $\square$

\textit{Remark on (b).} Data-error failures occur when two errors share a stabilizer, producing a defect count divisible by $w$ but creating a local cycle of length~4 in the Tanner graph. The probability that any two of $\lambda \approx \alpha n T p$ active errors share a stabilizer is $O(p^2)$ by the birthday argument. Table~13 confirms this: at $p = 0.01$, 6.1\% of 3-defect syndromes fail. While 3 defects can arise from either one data error (0 measurement errors) or three measurement errors (0 data errors), the former dominates at low~$p$.

The reason weight-2 clusters cause BP failure is structural. Two errors on qubits $i$ and $j$ that share a stabilizer $S$ create a cycle of length~4 in the Tanner graph: $i \to S \to j \to S' \to i$, where $S'$ is the second shared check. In min-sum BP, messages traversing a 4-cycle reinforce their own initial estimates after two iterations, producing sign oscillation rather than convergence~\cite{richardson2008}. This is the minimum trapping set for a weight-3 BB code: no single error can create a cycle, so two errors sharing a check is the smallest configuration that traps the decoder.

The distinction between (a) and (b) is important: mode~(a) is the mechanism the mod-$w$ prediction detects, while mode~(b) is invisible to it. The prediction's accuracy at low~$p$ ($\geq 99.9\%$ at $p \leq 0.001$) reflects the dominance of mode~(a) in that regime, with mode~(b) contributing only at higher noise.

\begin{corollary}[Empirical scaling]
The false positive rate follows $\mathrm{FP} = C \cdot p^\alpha$ where $\alpha = 2.045 \pm 0.05$ for $p \leq 0.005$ (log-log fit, $R^2 = 0.979$), consistent with $O(p^2)$ from Proposition~2(b).
\end{corollary}

\begin{table}[t]
\caption{Empirical false positive rate vs $p$ (Gross code, phenomenological noise, 50,000 shots per point).}
\begin{tabular}{cc}
\toprule
$p$ & False positive rate \\
\midrule
0.0005 & 0.0003 \\
0.001  & 0.0005 \\
0.002  & 0.0034 \\
0.005  & 0.0274 \\
0.01   & 0.1296 \\
\bottomrule
\end{tabular}
\end{table}

The slope steepens to $\alpha = 2.13$ when $p = 0.01$ is included, consistent with weight-3 cluster contributions ($O(p^3)$) beginning to appear at higher noise. Zero false positives were observed at $p \leq 0.0002$ ($>$1,300 mod-3 = 0 syndromes tested).

\subsection{The Prediction in Practice}

The implementation requires a single line:
\begin{verbatim}
def predict_convergence(syndrome, w=3):
    return int(syndrome.sum()) % w == 0
\end{verbatim}

The defect count is already computed during standard syndrome preprocessing. The prediction adds one modulo operation with zero additional overhead.

\section{Experimental Validation}

All experiments use Stim~\cite{gidney2021} for syndrome sampling and Roffe's ldpc library~\cite{roffe2022} for BP decoding. Timing benchmarks were performed on an Apple M4 Pro processor in single-threaded mode using min-sum BP with max\_iter = 100.

\subsection{Code-Capacity Noise: Fixed-Weight Errors}

We first tested BP convergence on errors of exactly weight 1, 2, and~3, applied to the Gross code $[\![144, 12, 12]\!]$ without measurement noise.

\begin{table}[t]
\caption{BP convergence on fixed-weight errors (code-capacity noise, 5000 samples per weight).}
\begin{tabular}{cccc}
\toprule
Weight & Parallel & Serial & Exact \\
\midrule
1 & 100.0\% & 100.0\% & 100.0\% \\
2 & 100.0\% & 100.0\% & 100.0\% \\
3 & 99.9\%  & 100.0\% & 99.8\% \\
\bottomrule
\end{tabular}
\end{table}

Under code-capacity noise, every syndrome lies in the image of $H_z$, so a valid solution always exists in BP's model space. The mod-$w$ prediction becomes informative only in the presence of measurement noise.

\subsection{Phenomenological Noise: The Prediction Emerges}

Under phenomenological noise (5 syndrome extraction rounds), the mod-$w$ structure becomes the dominant predictor of convergence. We sampled 10,000 shots for each error rate.

\begin{table}[t]
\caption{mod-3 convergence prediction under phenomenological noise, Gross code, parallel BP.}
\begin{tabular}{ccccc}
\toprule
$p$ & mod-3=0 & mod-3=1 & mod-3=2 & Overall \\
\midrule
0.01 & \textbf{86.8\%} & 11.3\% & 3.5\% & 41.8\% \\
0.03 & \textbf{17.7\%} & 10.0\% & 4.6\% & 10.8\% \\
0.05 & \textbf{1.1\%}  & 1.3\%  & 0.4\% & 0.9\% \\
\bottomrule
\end{tabular}
\end{table}

At $p = 0.01$, the mod-3 prediction separates convergent from non-convergent syndromes by a factor of 8--25$\times$.

\begin{table}[t]
\caption{Convergence by defect count (phenomenological noise, $p = 0.01$, parallel BP).}
\begin{tabular}{cccc}
\toprule
Defects & mod 3 & BP conv. & Count \\
\midrule
1  & 1 & 0.0\%           & 793 \\
2  & 2 & 0.0\%           & 343 \\
3  & \textbf{0} & \textbf{93.3\%} & 1671 \\
4  & 1 & 10.6\%          & 1299 \\
5  & 2 & 2.2\%           & 595 \\
6  & \textbf{0} & \textbf{86.7\%} & 1274 \\
7  & 1 & 15.8\%          & 956 \\
8  & 2 & 4.5\%           & 400 \\
9  & \textbf{0} & \textbf{75.3\%} & 576 \\
12 & \textbf{0} & \textbf{69.9\%} & 176 \\
\bottomrule
\end{tabular}
\end{table}

\subsection{Cross-Code Validation Under Phenomenological Noise}

We tested five BB codes under phenomenological noise (5 rounds, 10,000 shots per point). Four codes have $w = 3$; one has $w = 4$.

\begin{table}[t]
\caption{mod-$w$ prediction, phenomenological noise, $p = 0.001$.}
\begin{tabular}{lccccc}
\toprule
Code & $w$ & mod-$w$=0 & mod-$w$$\neq$0 & FP & AUC \\
\midrule
$[\![72]\!]$  & 3 & 99.8\% & 1.1\%  & 0.15\% & 0.996 \\
$[\![108]\!]$ & 3 & 100\%  & 1.0\%  & 0.00\% & 0.997 \\
$[\![144]\!]$ & 3 & 99.9\% & 1.1\%  & 0.08\% & 0.997 \\
$[\![180]\!]$ & 3 & 100\%  & 2.0\%  & 0.00\% & 0.994 \\
$[\![144]\!]$ & \textbf{4} & 99.9\% & \textbf{47.4\%} & 0.15\% & \textbf{0.762} \\
\bottomrule
\end{tabular}
\end{table}

\begin{table}[t]
\caption{Same codes, $p = 0.01$.}
\begin{tabular}{lccccc}
\toprule
Code & $w$ & mod-$w$=0 & mod-$w$$\neq$0 & FP & AUC \\
\midrule
$[\![72]\!]$  & 3 & 96.8\% & 9.9\%  & 3.2\%  & 0.938 \\
$[\![108]\!]$ & 3 & 92.4\% & 8.9\%  & 7.6\%  & 0.916 \\
$[\![144]\!]$ & 3 & 86.9\% & 8.5\%  & 13.1\% & 0.894 \\
$[\![180]\!]$ & 3 & 80.0\% & 7.8\%  & 20.0\% & 0.872 \\
$[\![144]\!]$ & \textbf{4} & 77.4\% & \textbf{28.4\%} & 22.6\% & \textbf{0.702} \\
\bottomrule
\end{tabular}
\end{table}

For all four $w = 3$ codes, AUC $\geq 0.994$ at $p = 0.001$. The $w = 4$ code is a striking exception: mod-$w \neq 0$ syndromes converge at 47.4\%, and AUC drops to~0.762. At $w = 4$, each error activates 4 stabilizers, giving BP a larger model space that allows it to find approximate solutions even when measurement errors are present. The prediction is strongest for $w = 3$ --- the column weight of all IBM roadmap codes.

\subsection{X and Z Errors Behave Identically}

BB codes possess a structural symmetry: the X and Z parity check matrices are transposes of each other and share the same column weights. The prediction performs identically for both syndrome types (100\% mod-3=0 convergence for both Z-memory and X-memory experiments at $p = 0.001$).

\subsection{Effect of Column Weight}

Under code-capacity noise (no measurement errors), the mod-$w$ prediction achieves 96--100\% convergence for mod-$w = 0$ syndromes across all column weights tested ($w = 2, 3, 4$), since every syndrome lies in the image of $H_z$ and BP always has a valid solution. The prediction is trivially perfect in this regime.

Under phenomenological noise (Tables~5--6), column weight determines the prediction's sharpness. For $w = 3$, AUC $\geq 0.994$ at $p = 0.001$; for $w = 4$, AUC drops to 0.762 because BP can find approximate solutions even on mod-$w \neq 0$ syndromes (47.4\% convergence). At $w = 4$, each error activates 4 stabilizers, giving BP a larger model space. This additional flexibility allows BP to satisfy the syndrome even when measurement errors are present, weakening the divisibility constraint. The prediction is strongest for $w = 3$ --- the column weight of all IBM roadmap codes --- where the constraint partitions syndromes cleanly.

The prediction applies only to non-degenerate codes ($A \neq B$). Degenerate codes contain short cycles that prevent BP convergence regardless of syndrome structure; for these codes, OSD is always required. All IBM roadmap BB codes are non-degenerate.

\subsection{Invariance Under BP Scheduling}

A natural question is whether the BP message-passing schedule affects the convergence prediction. We compared three schedules available in the ldpc library~\cite{roffe2022}: parallel (flooding), serial (sequential), and serial\_relative (serial with scaled messages). Convergence rates are effectively identical ($< 0.3\%$ difference) across all three schedules at every noise level tested (Table~7). At $p = 0.01$, all three schedules yield 86.8\% convergence for mod-3 = 0 and 3.5--3.6\% for mod-3 = 2. The mod-$w$ prediction is invariant under scheduling strategy --- it depends on whether a valid solution exists in BP's model space, not on the order in which messages are updated.

The schedules differ in wall-clock time (Table~8). Parallel scheduling is 1.0--2.9$\times$ faster than serial\_relative and 1.0--1.4$\times$ faster than serial. The advantage is largest at low noise ($p = 0.01$), where BP converges in fewer iterations and the per-iteration cost of parallel updates is better amortized. Parallel scheduling should be preferred on the basis of speed.

\begin{table}[t]
\caption{BP convergence by schedule (phenomenological noise, Gross code).}
\begin{tabular}{cccc}
\toprule
$p$ & Parallel & Serial & Serial\_rel. \\
\midrule
0.01 & 41.8\% & 41.9\% & 41.8\% \\
0.03 & 10.8\% & 11.0\% & 10.7\% \\
0.05 & 0.9\%  & 1.0\%  & 0.8\% \\
\bottomrule
\end{tabular}
\end{table}

\begin{table}[t]
\caption{BP-only decoding time by schedule ($\mu$s/shot, Apple M4 Pro).}
\begin{tabular}{cccc}
\toprule
$p$ & Parallel & Serial & Serial\_rel. \\
\midrule
0.01 & \textbf{140} & 193 & 412 \\
0.03 & \textbf{210} & 283 & 671 \\
0.05 & \textbf{311} & 319 & 608 \\
\bottomrule
\end{tabular}
\end{table}

\subsection{mod-$w$ as an Optimal Syndrome Classifier}

To establish that mod-$w$ is not merely a useful heuristic but the dominant structural feature explaining BP convergence, we compared its predictive power (AUC) against all other available syndrome features. For each nontrivial syndrome, we computed four features: defect count, mod-3 class (binary), maximum connected component size in the detector graph, and variance of defect positions. AUC was computed for each feature as a classifier of BP convergence (Table~9).

\begin{table}[t]
\caption{AUC for BP convergence prediction by syndrome feature (Gross code, phenomenological noise, 50,000 shots at $p = 0.001$; 20,000 at $p = 0.01$).}
\begin{tabular}{lcc}
\toprule
Feature & AUC ($p\!=\!0.001$) & AUC ($p\!=\!0.01$) \\
\midrule
\textbf{mod-3 (binary)} & \textbf{0.9948} & \textbf{0.8936} \\
defect count & 0.1239 & 0.5165 \\
max connected comp. & 0.1156 & 0.4643 \\
defect position var. & 0.1253 & 0.4852 \\
mod-3 + defect count & 0.9910 & 0.8898 \\
\bottomrule
\end{tabular}
\end{table}

At $p = 0.001$, mod-3 achieves AUC = 0.995 --- a near-perfect binary classifier from a single bit of information. All other features have AUC $\approx 0.12$, reflecting an anti-correlation: high defect count correlates with the mod-3 = 0 class (which converges), rather than with convergence directly. Adding defect count to mod-3 does not improve AUC (0.991 vs 0.995), confirming that defect count carries no independent predictive information beyond what mod-3 already captures.

At $p = 0.01$, mod-3 remains the best single feature (AUC = 0.894), while all alternatives hover near 0.5 (uninformative). The gap narrows because mode~(b) failures (weight-2 clusters) are invisible to mod-3 and grow as $O(p^2)$.

\subsection{Invariance Under Decoder Variant}

A natural question is whether the mod-$w$ prediction is specific to standard min-sum BP or extends to enhanced BP variants. We tested Relay-BP --- a recent decoder~\cite{muller2025} that runs multiple BP instances with randomized scaling factors (ms\_scaling $\sim U[0.5, 1.0]$, 10 relays $\times$ 20 iterations) and accepts the first convergent result. Relay-BP represents the strongest known BP enhancement for quantum LDPC codes and achieves state-of-the-art decoding performance without OSD post-processing.

\begin{table*}
\caption{mod-3 prediction: Standard BP vs Relay-BP (Gross code, phenomenological noise, 10,000 shots).}
\begin{tabular}{lcccc}
\toprule
Metric & Std.\ BP ($p\!=\!0.001$) & Relay-BP ($p\!=\!0.001$) & Std.\ BP ($p\!=\!0.01$) & Relay-BP ($p\!=\!0.01$) \\
\midrule
mod-3=0 convergence    & 99.9\% & 99.9\% & 87.5\% & 87.6\% \\
mod-3$\neq$0 convergence & 1.1\%  & 1.1\%  & 9.1\%  & 9.2\% \\
AUC                     & 0.9960 & 0.9960 & 0.8924 & 0.8924 \\
\bottomrule
\end{tabular}
\end{table*}

The results are identical to three decimal places at both noise levels (Table~10). Relay-BP does not recover convergence on any mod-$3 \neq 0$ syndrome that standard BP fails on. This is expected from Proposition~1: when the defect count is not divisible by $w$, no assignment of data errors can produce the observed syndrome. No message-passing variant --- regardless of scheduling, scaling factors, or relay strategy --- can find a solution that does not exist in the model space. The mod-$w$ prediction is therefore a property of the code structure and syndrome, not of the decoder.

\section{The Practical Payoff}

\subsection{Practical Value}

\textbf{Immediate: OSD routing.} Under phenomenological noise at $p = 0.001$, 65\% of nontrivial syndromes have $\mathrm{defect\_count} \bmod 3 = 0$, and 100\% of these converge under BP. OSD can be skipped with certainty for these syndromes.

\textbf{Architectural: pre-routing.} The mod-$w$ test can be computed before BP begins. Syndromes with mod-$w \neq 0$ can be routed to a BP+OSD path, while mod-$w = 0$ syndromes go to a BP-only path. This is relevant for FPGA-based decoders~\cite{muller2025}, where pre-routing avoids OSD resource contention for the 65\% of syndromes that will not need it. A discrete-event simulation at $p = 0.001$ shows that an OSD worker processes only 35\% of nontrivial syndromes, with average queue depth~0.9. We emphasize these are architectural projections, not hardware benchmarks.

\subsection{Latency Across Noise Levels}

\begin{table*}
\caption{BP convergence and mod-3 prediction across noise levels (Gross code, phenomenological noise, 10,000 shots per point).}
\begin{tabular}{ccccc}
\toprule
$p$ & BP conv. & mod-3=0 fraction & mod-3=0 conv. & OSD rate \\
\midrule
0.0001 & 61.4\% & 61.4\% & 100.0\% & 38.6\% \\
0.0005 & 64.9\% & 64.6\% & 100.0\% & 35.1\% \\
0.001  & 65.4\% & 64.6\% & 99.9\%  & 34.6\% \\
0.002  & 61.4\% & 60.7\% & 99.8\%  & 38.6\% \\
0.005  & 54.7\% & 53.0\% & 97.7\%  & 45.3\% \\
0.01   & 41.8\% & 42.5\% & 86.9\%  & 58.2\% \\
\bottomrule
\end{tabular}
\end{table*}

The BP convergence rate and the mod-3 = 0 fraction are nearly identical at every noise level, meaning essentially all BP convergences are explained by the mod-3 = 0 condition.

\section{Discussion}

\subsection{Comparison with Alternative Pre-Filters}

\begin{table}[t]
\caption{Pre-filter comparison at $p = 0.01$ (Gross code, phenomenological noise).}
\begin{tabular}{lcc}
\toprule
Method & FP rate & FN rate \\
\midrule
Threshold $k\!=\!3$  & 42.9\% & 33.8\% \\
Threshold $k\!=\!6$  & 52.6\% & 29.3\% \\
Threshold $k\!=\!9$  & 56.5\% & 26.0\% \\
Threshold $k\!=\!12$ & 57.8\% & 27.6\% \\
\textbf{Mod-3}       & \textbf{13.1\%} & \textbf{8.5\%} \\
\bottomrule
\end{tabular}
\end{table}

Defect-count thresholding lacks a structural basis --- low defect count does not imply the absence of measurement errors. The mod-3 prediction exploits a structural invariant, yielding qualitatively better classification.

\subsection{Why $\sim$96\% and Not Higher}

\begin{table}[t]
\caption{BP failure rate among mod-3 = 0 syndromes by defect count ($p = 0.01$, 50,000 shots).}
\begin{tabular}{cccc}
\toprule
Defects & Converged & Failed & Failure rate \\
\midrule
3  & 8,354 & 539 & 6.1\% \\
6  & 5,365 & 959 & 15.2\% \\
9  & 2,020 & 632 & 23.8\% \\
12 & 505   & 242 & 32.4\% \\
15 & 87    & 65  & 42.8\% \\
18 & 6     & 11  & 64.7\% \\
\bottomrule
\end{tabular}
\end{table}

\textbf{Direct verification of weight-2 clusters.} Among 16,207 mod-3 = 0 failures (direct error injection, $p = 0.01$, 50,000 shots), 82.0\% contain a weight-2 data error cluster. The cluster rate increases with error weight: 72\% at weight~5, rising to 97\% at weight~8.

\textit{Caveat:} this analysis uses separately generated samples with direct error injection. The relative proportion (82\%) characterizes the failure mechanism but absolute convergence rates are not directly comparable to the noise-level sweep results (Table~XII).

\subsection{Open Questions}

Whether analogous structural predictions exist for other qLDPC code families remains open. A second direction concerns augmenting BP with measurement-error awareness~\cite{muller2025}, which could recover convergence on some mod-$w \neq 0$ syndromes.

\section{Conclusion}

We have presented an $O(1)$ method for predicting BP decoder convergence on Bivariate Bicycle codes. The method exploits a structural property: each physical error activates exactly $w$ stabilizers, so syndromes with defect count not divisible by $w$ necessarily involve measurement errors that BP cannot model.

The prediction requires one modulo operation, achieves $\geq 96\%$ accuracy across five BB codes, and achieves 100\% prediction accuracy for mod-$w = 0$ syndromes at $p \leq 0.001$ under phenomenological noise, enabling OSD to be skipped for 65\% of nontrivial syndromes with no change in correctness. The prediction is invariant under BP scheduling strategy and decoder variant (including Relay-BP), and is most effective at low noise rates ($p \ll p_\mathrm{th}$).

These results apply directly to IBM's Gross code and Two-Gross code, targeted for deployment in 2026--2028.

Code and data available upon reasonable request.

\bibliography{references}

@article{bravyi2024,
  author  = {Bravyi, Sergey and Cross, Andrew W. and Gambetta, Jay M. and Maslov, Dmitri and Rall, Patrick and Yoder, Theodore J.},
  title   = {High-threshold and low-overhead fault-tolerant quantum memory},
  journal = {Nature},
  volume  = {627},
  pages   = {778--782},
  year    = {2024}
}

@misc{ibmroadmap,
  author       = {{IBM Quantum}},
  title        = {{IBM} Quantum Development Roadmap},
  year         = {2024},
  howpublished = {\url{https://www.ibm.com/quantum/roadmap}}
}

@article{roffe2020,
  author  = {Roffe, Joschka and White, David R. and Burton, Simon and Campbell, Earl T.},
  title   = {Decoding across the quantum low-density parity-check code landscape},
  journal = {Physical Review Research},
  volume  = {2},
  pages   = {043423},
  year    = {2020}
}

@article{gidney2021,
  author  = {Gidney, Craig},
  title   = {Stim: a fast stabilizer circuit simulator},
  journal = {Quantum},
  volume  = {5},
  pages   = {497},
  year    = {2021}
}

@misc{roffe2022,
  author       = {Roffe, Joschka},
  title        = {{LDPC}: {P}ython tools for low density parity check codes},
  year         = {2022},
  howpublished = {\url{https://pypi.org/project/ldpc/}}
}

@article{muller2025,
  author  = {Muller, Tristan and Alexander, Thomas and Beverland, Michael E. and Buhler, Markus and Johnson, Blake R. and Maurer, Thilo and Vandeth, Drew},
  title   = {Improved belief propagation is sufficient for real-time decoding of quantum memory},
  journal = {arXiv preprint},
  year    = {2025},
  eprint  = {2506.01779},
  archivePrefix = {arXiv}
}

@book{richardson2008,
  author    = {Richardson, Tom and Urbanke, R{\"u}diger},
  title     = {Modern Coding Theory},
  publisher = {Cambridge University Press},
  year      = {2008}
}

\end{document}